\documentclass[article,onecolumn]{IEEEtran}
\usepackage{graphicx,colordvi,psfrag,bbm}
\usepackage{amsmath,amssymb}
\usepackage{epstopdf}
\usepackage[caption=false]{subfig}
\usepackage{epsfig,cite}
\usepackage{calc,pstricks, pgf, xcolor}
\usepackage{nicefrac}
\usepackage{enumerate}
\usepackage{setspace}
\usepackage{bm}
\usepackage{dsfont}

\newtheorem{theorem}{Theorem}
\newtheorem{corollary}{Corollary}

\newtheorem{example}{Example}

\newtheorem{lemma}{Lemma}
\newtheorem{proposition}{Proposition}
\newenvironment{proof}[1][Proof]{\noindent\textbf{#1.} }{\ \rule{0.5em}{0.5em}}

\newcommand{\by}{\mathbf{y}}
\newcommand{\bY}{\mathbf{Y}}
\newcommand{\bx}{\mathbf{x}}
\newcommand{\bX}{\mathbf{X}}

\newcommand{\bZ}{\mathbf{Z}}

\newcommand{\bu}{\mathbf{u}}

\newcommand{\bw}{\mathbf{w}}
\newcommand{\bW}{\mathbf{W}}

\newcommand{\RR}{\mathbb{R}}

\newcommand{\Unif}{\mathop{\mathrm{Unif}}}

\newcommand{\Ind}{\mathds{1}(Y=A(X))}
\newcommand{\adj}{\sim}

\newcommand{\Indnonempty}{\mathds{1}(X\sim Y)}

\newcommand{\Indy}{\mathds{1}(y=A(X))}
\newcommand{\Expt}{\mathbb{E}}

\newcommand{\ind}{\mathds{1}(y=a(x))}
\newcommand{\Nxy}{N(x,y)}
\newcommand{\indnonempty}{\mathds{1}(x\sim y)}
\newcommand{\indnonemptyy}{\mathds{1}(X\sim y)}
\newcommand{\indinD}{\mathds{1}(a\in\Actset(\bx,\by))}
\newcommand{\I}{\mathds{1}}

\newcommand{\Actset}{\mathcal{A}}

\newcommand{\Tx}{\mathcal{S}}

\newcommand{\m}[1]
{\mathcal{#1}}

\DeclareMathOperator*{\argmin}{\arg\!\min}

\newrgbcolor{BoxRed}{1 .5 .5}
\newrgbcolor{BoxBlue}{.9 .9 1}
\newrgbcolor{lightgrey}{.7 .7 0.7}

\def\dotleq{\stackrel{.}{\leq}}

\begin{document}

\title{Mutual Information Bounds via Adjacency Events}

\author{Yanjun Han, Or~Ordentlich,
        and Ofer~Shayevitz,~\IEEEmembership{Senior Member,~IEEE}
\thanks{
The work of Y. Han was supported in part by the NSF Center for Science of Information under grant agreement CCF-0939370. The work of O. Ordentlich was supported by the MIT - Technion postdoctoral fellowship. The work of O. Shayevitz was supported by an ERC grant no. 639573, and an ISF grant no. 1367/14.
The material in this paper was presented in part at the International Symposium on Information
Theory (ISIT) 2014 in Honolulu, Hawaii, USA}}


\parskip 3pt

\maketitle

\begin{abstract}
The mutual information between two jointly distributed random variables $X$ and $Y$ is a functional of the joint distribution $P_{XY}$, which is sometimes difficult to handle or estimate. A coarser description of the statistical behavior of $(X,Y)$ is given by the marginal distributions $P_X, P_Y$ and the \textit{adjacency} relation induced by the joint distribution, where $x$ and $y$ are adjacent if $P(x,y)>0$. We derive a lower bound on the mutual information in terms of these entities. The bound is obtained by viewing the channel from $X$ to $Y$ as a probability distribution on a set of possible \textit{actions}, where an action determines the output for any possible input, and is independently drawn. We also provide an alternative proof based on convex optimization, that yields a generally tighter bound. Finally, we derive an upper bound on the mutual information in terms of adjacency events between the action and the pair $(X,Y)$, where in this case an action $a$ and a pair $(x,y)$ are adjacent if $y=a(x)$. As an example, we apply our bounds to the binary deletion channel and show that for the special case of an i.i.d. input distribution and a range of deletion probabilities, our lower and upper bounds both outperform the best known bounds for the mutual information.
\end{abstract}

\section{Introduction}
\label{sec:intro}
The mutual information $I(X;Y)$ between two jointly distributed random variables $X$ and $Y$ arises as the fundamental limit in many information theoretic problems. When the alphabets $\m{X}$ and $\m{Y}$ are small, the computation of $I(X;Y)$ can be performed directly. This is the typical scenario when considering e.g. the calculation of capacity of memoryless channels, assuming the optimal input distribution is known. In many cases however, the alphabet may become large or even grow unbounded; this is the case e.g. with the capacity of channels with memory that are information stable \cite{dobrushin63}, where the capacity is essentially given by the limit of $I(X^n;Y^n)/n$, for the optimal input $X^n$. In such cases, it often becomes prohibitively difficult or even virtually impossible to precisely compute the mutual information, hence one must resort to bounding techniques.

In many problems, the marginal distributions of $X$ and $Y$ are simple and the computation of the entropies $H(X)$ and $H(Y)$ is more tractable. In such cases the main obstacle becomes handling the joint distribution and computing the joint (or conditional) entropy. One such prominent example is the binary deletion channel \cite{mitzenmacher09} with deletion probability $d$  and an i.i.d. uniform input process. For this setting, the normalized output entropy is easy to derive and approaches $(1-d)$. However, to evaluate the joint distribution for any given input-output pair, one needs to find the number of different ways the output can be obtained from the input by deleting input bits. This is a difficult combinatorial question, and consequently computing the joint entropy is very challenging. A simpler combinatorial question is to determine whether the output can be obtained from the input by \textit{some} deletion pattern. More generally put, instead of fully characterizing the joint distribution, it is sometimes much easier to characterize its support. Thus, the goal of this work is to provide bounds on the mutual information as a function of the marginals and the joint support. These bounds will be useful when the support is sparse.

In what follows, we assume the alphabets $\m{X},\m{Y}$ are finite unless otherwise stated. We say that $x$ and $y$ are \textit{adjacent} if $P_{XY}(x,y) > 0$, and we denote this relation by $x\adj y$. We call the event $\I(x\adj y)$ an \textit{adjacency event}. Our first main result is the following.
\begin{theorem}
  For any jointly distributed discrete r.vs $(X,Y)$,
  \begin{align}\label{eq:ourbound}
    I(X,Y) \geq -\Expt_Y\log\Expt_X\I(X\adj Y) - \Expt_X\log\Expt_Y\frac{\I(X\adj Y)}{\Expt_X\I(X\adj Y)}
  \end{align}
\label{thm:main}
\end{theorem}

Note that by Jensen's inequality both summands are non-negative, and therefore as a corollary we also get that $I(X,Y) \geq -\Expt_Y\log\Expt_X\I(X\adj Y)$. One can find examples where both bounds are tight, e.g., for the mutual information between input and output of the binary erasure channel. It is instructive to note that the weaker bound can be derived directly by the following argument. Draw an i.i.d. codebook with block length $n$ according to $P_X$, and use it to communicate over a memoryless channel $P_{Y|X}$. Consider the following decoding rule: If the output sequence $y^n$ is $P_Y$-typical and there is a \textit{unique} codeword $x^n$ such that $x_k\adj y_k$ for all $k$, output that codeword. Otherwise, declare an error. Clearly, $\Pr(X_k\adj y_k) = E_X \I(X\sim y_k)$ and thus, assuming that $y^n$ is typical, the probability (averaged over random codebooks) that a specific codeword will satisfy the decoding rule is $\approx \prod_{y\in\mathcal{Y}}\left(\Expt_X \I(X\sim y)\right)^{nP(y)}$. Therefore, by the union bound, any rate below $-\Expt_Y\log\Expt_X\I(X\adj Y)$ can be attained by this strategy with vanishing error probability, and this in turn cannot be larger than the mutual information. A bound of this type was implicitly used in~\cite{dg01,dm06}. Our main contribution is therefore the second term in~\eqref{eq:ourbound}. As we shall see in Section \ref{sec:examples}, this additional term can be significant.

Let us briefly provide the main ideas behind our approach. A channel is traditionally defined via a conditional probability distribution $P_{Y|X}$ of the output given the input. Alternatively, a channel can also be (nonuniquely) defined as a random mapping $Y=A(X)$ from an input alphabet to an output alphabet, where the actual mapping applied to the input, namely the channel \textit{action} $A$, is drawn according to some probability distribution $P_A$ over the set of all possible actions, \textit{independently} of the input (see the functional representation lemma in \cite[Appendix B]{ElGamalKimBook}). Following this paradigm, the mutual information for a given input distribution $P_X$ can be written as
\begin{align}
I(X;Y)&=H(Y) - H(Y|X) \nonumber\\
&=H(Y) - (H(Y,A|X) - H(A|X,Y)) \nonumber\\
&=H(Y) - H(A|X) - H(Y|A,X) + H(A|X,Y)\nonumber\\
&=H(Y) - H(A) + H(A|X,Y)\label{cleancapacity}
\end{align}
where~\eqref{cleancapacity} follows since the action $A$ is statistically independent of the input $X$, and $Y=A(X)$. This holds for any eligible choice of action $A$. A natural quantity to consider is therefore the \textit{intrinsic uncertainty} $H(A|X,Y)$ associated with $A$, that captures the amount of information regarding the channel action revealed by observing its input and output. Note that for any eligible choice of $A$, we have that $I(A;X,Y) = H(A) - H(A|X,Y) = H(Y|X)$ is fixed, but the entropy of the action $H(A)$ and the intrinsic uncertainty associated with the action can vary.

As an example, consider the binary symmetric channel (BSC) with crossover probability $0<p<\tfrac{1}{2}$. A natural choice for the action $A$ is drawing a r.v. $Z\sim \textrm{Bern}(p)$ and setting $A(X) = X \oplus Z$. In this case, the entropy of the action is $H(A)=h(p)$, where $h(\cdot)$ is the binary entropy function, and the intrinsic uncertainty $H(A|X,Y)=0$, since viewing $X$ and $Y$ completely reveals the action (the noise $Z$). Another possible choice for the action $A$ is drawing a ternary r.v. $U$ with $\Pr(U=0)=\Pr(U=1)=p$, and $\Pr(U=2) = 1-2p$, and setting
\begin{align*}
A(X) = U\cdot\I(U\neq 2) + X\cdot\I(U = 2)
\end{align*}
In this case, the entropy of the action is $H(A) = h(2p) + 2p$, and the intrinsic uncertainty is $H(A|X,Y) = (1-p)\cdot h\left(\tfrac{p}{1-p}\right) > 0$, since if $X=Y$ there remains some uncertainty regarding the action. Indeed, it can be directly verified that the identity $h(p) = h(2p) + 2p - (1-p)\cdot h\left(\tfrac{p}{1-p}\right)$ holds.

Following the above, in Section \ref{sec:family-bounds-via} we derive a lower bound on the intrinsic uncertainty for any given choice of the action $A$. This bound is based on an application of the Donsker-Varadhan variational principle. This will immediately translate into  lower bounds on the mutual information. Our general statement, given in Theorem \ref{thm:genLB}, is a family of bounds that depend on the particular choice of the action. While these bounds may be generally difficult to evaluate, we show in Section \ref{sec:bound-via-adjacency} that for any channel $P_{Y|X}$ there always exists a specific choice of action, such that the associated bound depends only on the marginals and the joint support. This yields Theorem \ref{thm:main}.

The proof of Theorem \ref{thm:main} as delineated above in based on information theoretic arguments. Alternatively, the theorem can also be proved more directly using convex optimization techniques. In fact, this alternative approach does not only recover Theorem \ref{thm:main}, but can also yield an increasing sequence of bounds that converges to the best possible lower bound on the mutual information in terms of the marginals $P_X,P_Y$ and the support of $P_{XY}$. Furthermore, while the information theoretic proof applies only to finite alphabets, the convex optimization approach can also handle countably infinite alphabets. This result appears in Theorem \ref{thm:convex_opt}, Section \ref{sec:conv}. We note that the improved bounds obtained by this procedure seem quite difficult to evaluate in general.

Interestingly, while actions were introduced in order to lower bound the mutual information, our results can be trivially leveraged to obtain upper bounds as well.
\begin{theorem}\label{thm:upper}
  Let  $(X,Y)\sim P_X\times P_{Y|X}$ be jointly distributed discrete r.vs. Let $A$ be any action consistent with $P_{Y|X}$, i.e., such that $Y=A(X)$. Then
\begin{align}\label{eq:upper}
I(X; Y)\leq H(Y) + \Expt_A\log \Expt_{X,Y} \I(A\sim (X,Y)) + \Expt_{X,Y}\log\Expt_A\frac{\I(A\sim (X,Y))}{\Expt_{X,Y}\I(A\sim (X,Y))}
\end{align}
\end{theorem}
\begin{proof}
  By~\eqref{cleancapacity} we have that $I(X;Y) = H(Y) - I(A;X,Y)$. The proof follows by applying Theorem \ref{thm:main} to $I(A;X,Y)$.
\end{proof}

Note that $\I(A\adj (X,Y))$ is an indicator on the event where $A(X)=Y$. In~\eqref{eq:upper}, the expectations are taken with respect to $(X,Y)\sim P_{XY}$ and $A\sim P_A$ independent of $(X,Y)$. Observe also that both the second and third terms in~\eqref{eq:upper} are non-positive, hence the bound holds even if one of them is removed.

Lastly, in Section \ref{sec:examples} we illustrate the applicability of our bounds in several specific examples. In particular, we provide simple examples showing that our bounds are sometimes tight, and demonstrating that the second term in~\eqref{eq:ourbound} can be significant. We then consider the binary deletion channel for which the value of the mutual information is currently unknown for any nontrivial input process. For an i.i.d. uniform input, we evaluate our lower and upper bounds, and show that they both outperform the best known bounds on the mutual information. Finally, we draw a relation between the upper bound from Theorem~\ref{thm:upper} and a recent conjecture of Courtade and Kumar~\cite{ck14}. As all examples we consider in this paper involve binary channels, unless stated otherwise, all logarithms are taken to base $2$.

\section{A Family of Bounds via Actions}\label{sec:family-bounds-via}

In this section we define a channel by its action on its input, and develop general lower bounds on the mutual information between the input and output in terms of the channel action, by bounding the associated intrinsic uncertainty defined below.


\subsection{Channels via Actions}
Let $\mathcal{X,Y}$ be discrete alphabets. Any channel $P_{Y|X}$ from $\mathcal{X}$ to $\mathcal{Y}$ can be (nonuniquely) defined by a probability distribution $P_A$ on a set $\Actset$ of mappings from $\mathcal{X}\mapsto\mathcal{Y}$, to which we refer to below as \emph{actions}. Each action $a(\cdot)\in\Actset$ is defined for all possible inputs, and the channel action is chosen independently of the input, yielding the output $Y=A(X)\in\mathcal{Y}$.

For any eligible choice of action $A$, the \textit{intrinsic uncertainty} of the channel with respect to the input distribution $P_X$ is defined to be $H(A|X,Y)$. Note that while the intrinsic uncertainty may depend on the choice of $A$, the difference $H(A) - H(A|X,Y)$, which was shown in Section~\ref{sec:intro} to be equal to $H(Y|X)$, does not; we therefore have the freedom to choose the action distribution that is most convenient to work with.

\begin{example}[Generic action set]
For any channel $P_{Y|X}$ we can always generate the action according to the following procedure. Let $\m{A}$ consist of all $|\m{Y}|^{|\m{X}|}$ functions from $\m{X}\mapsto\m{Y}$, and for any $a\in\m{A}$ set $P_A(a)=\prod_{x\in\m{X}}P_{Y|X}(a(x)|x)$. Drawing $A$ according to $P_A$, statistically independent of $X$, and setting $Y=A(X)$, is equivalent to drawing in advance a sequence of statistically independent r.vs $\{Y_x\}_{x\in \m{X}}$, where $Y_x\sim P_{Y|X}(\cdot|x)$, and then when $X$ is revealed, outputting only the corresponding $Y_X$. Thus, the above $\m{A}$ and $P_A$ are consistent with $P_{Y|X}$, i.e., they describe the channel $P_{Y|X}$.
\label{ex:genericaction}
\end{example}

We further note that it is always possible to construct an action set with less than $|\m{X}|\cdot|\m{Y}|$ actions, see the functional representation lemma in~\cite[Appendix B]{ElGamalKimBook}. Moreover, in many cases there exist ``natural'' choices of an action that describes the channel. In Section \ref{sec:intro} we described such choices for the BSC. Below we provide a few more examples.

\begin{example}[Z Channel]
The (symmetric) Z channel has a binary input $X$ and binary output $Y$, such that $\Pr(Y=0|X=0)=1$ and $\Pr(Y=0|X=1)=\Pr(Y=1|X=1)=\tfrac{1}{2}$. A natural choice for the action $A$ is taking the action set $\m{A}$ to consist of the two actions $a_1(x)=x$ and $a_2(x)=0$ with probability assignment $p(a_1)=p(a_2)=\tfrac{1}{2}$.
\end{example}

\begin{example}[Deletion Channel]\label{ex:del}
In a deletion channel, each transmitted symbol is either deleted or received uncorrupted. Assuming the input to the channel is an $n$-dimensional vector $\bX$, the set $\Actset$ includes $2^n$ actions, each corresponding to a different subset of the input indices $[1:n]$ marked for deletion. In an i.i.d. deletion model symbols are independently deleted with probability $d$. Therefore the probability of an action $a$ that deletes exactly $w$ bits is $P(a) = d^w(1-d)^{n-w}$. Different actions applied to the same input may result in the same output. For example, if $\bx=01100$ we may get the output $\by=110$ if either the first and fourth symbols or the first and fifth symbols were deleted. Therefore, the intrinsic uncertainty $H(A|\bX,\bY)$ is generally positive.
\end{example}

\vspace{1mm}

\begin{example}[Trapdoor Channel]
The trapdoor channel is a simple finite-state binary channel, defined as follows. Balls labeled ``0'' or ``1'' are used to communicate through the channel. The channel starts with a ball already in it, referred to as the initial state. On each channel use, a ball is inserted into the channel by the transmitter, and one of the two balls in the channel is emitted with equal probability. The ball that is not emitted remains inside for the next channel use. In this model, the channel's action consists of choosing the initial state and deciding for each channel use whether to emit the ball that was already inside the channel or the ball that has just entered. Since an input $\bx$ can be mapped to an output $\by$ via multiple actions, the intrinsic uncertainty is generally positive.
\end{example}

\subsection{Bounds}

Our main tool in lower bounding the intrinsic uncertainty is the variational formula of Donsker and Varadhan (See, e.g.,~\cite[Chapter 1.4]{Dupuis_and_Ellis}). We write $D(P\|Q)$ for the relative entropy between the distributions $P,Q$, and $Q\ll P$ if $P(x)=0$ implies $Q(x)=0$.

\begin{lemma}[Donsker-Varadhan]
For any distribution $P$ and any nonnegative function $f(x)$ for which $\mathbb{E}_P\log f(X)$ is finite,
\begin{align}
\mathbb{E}_P\log f(X)=\min_{Q\ll P}\log\mathbb{E}_Q f(X)+D(P\|Q),\label{laplace}
\end{align}
and the minimum is uniquely attained by
\begin{align}
Q^*(x)=\frac{P(x)/f(x)}{\mathbb{E}_P (1/f(X))},
\end{align}
where by convention we set $1/f(x)=0$ if $f(x)=0$.
\end{lemma}

For completeness, we bring the proof of this lemma.

\begin{proof}
Let $Q^*(x)$ be as above. For any distribution $Q$ we have
\begin{align}
D(P\|Q)+\log\mathbb{E}_Q f(X)&=\mathbb{E}_P\log\frac{P}{Q}+\log\mathbb{E}_Q f(X)\nonumber\\
&=\mathbb{E}_P\log\frac{Q^*}{Q}+\mathbb{E}_P\log\frac{P}{Q^*}+\log\mathbb{E}_Q f(X)\nonumber\\
&={\sum_{x} P(x)\log\frac{P(x)}{Q(x)f(x)\mathbb{E}_P(1/f(X))}+\mathbb{E}_P\log\frac{P(X)f(X)\mathbb{E}_P(1/f(X))}{P(X)}+\log\mathbb{E}_Q f(X)}\nonumber\\
&\stackrel{(a)}{\geq} {\left(\sum_x P(x)\right)\log\frac{\sum_x P(x)}{\sum_x Q(x)f(x)\mathbb{E}_P(1/f(X))}+\mathbb{E}_P \log f(X) +\log\mathbb{E}_P\frac{1}{f(X)}+\log\mathbb{E}_Q f(X)}\nonumber\\
&=\mathbb{E}_P \log f(X)\nonumber
\end{align}
where $(a)$ follows from the log-sum inequality~\cite[Chapter 2.7]{cover} which is tight if and only if $Q(x)=Q^*(x)$.
\end{proof}

We would like to obtain an alternative expression for
\begin{align}
H(A|X,Y)=\mathbb{E}\log \frac{1}{P(A|X,Y)},\label{GeneralEntropy}
\end{align}
where the expectation is taken with respect to the joint distribution
\begin{align}
P(x,y,a)&=P(x)P(a|x)P(y|x,a)\nonumber\\
&=P(x)P(a)\ind,\nonumber
\end{align}
and $\mathds{1}(B)$ is an indicator function for the event $B$. For brevity, we sometimes refer to this distribution as $P$.

Define the distribution
\begin{align}
Q(x,y,a)\triangleq\frac{P(x,y,a)P(a|x,y)}{\mathbb{E}_P P(A|X,Y)},\label{Qdist}
\end{align}
which we sometimes refer to as $Q$. Using the Donsker-Varadhan variational principle with $f(x,y,a)=1/P(a|x,y)$, the expectation from~\eqref{GeneralEntropy} can be written as
\begin{align}
\mathbb{E}\log \frac{1}{P(A|X,Y)}&=\log\mathbb{E}_Q \frac{1}{P(A|X,Y)}+D(P\|Q)\nonumber\\
&=\log\mathbb{E}_Q \frac{1}{P(A|X,Y)}+D\left(P_Y\|Q_Y\right)+D\left(P_{X,A|Y}\|Q_{X,A|Y} \mid P_Y\right),\label{LaplaceApp}
\end{align}
where~\eqref{LaplaceApp} follows from the chain rule of relative entropy. The marginal distribution $Q(y)$ is given by
\begin{align}
  Q(y)&=\sum_{x,a}Q(x,y,a)\nonumber\\
&=\frac{1}{\mathbb{E}_P P(A|X,Y)}\sum_{x,a}P(x)P(a)\ind P(a|x,y)\nonumber\\
&=\frac{\mathbb{E}_{X,A} P(A|X,y)}{\mathbb{E}_P P(A|X,Y)},\label{Qy}
\end{align}
where in~\eqref{Qy} we have used the fact that $P(a|x,y)=0$ whenever $y\neq a(x)$. Thus,
\begin{align}
D\left(P_Y\|Q_Y\right)&=\mathbb{E}_{Y}\log\left(\frac{P(Y)\mathbb{E}_P P(A|X,Y)}{\mathbb{E}_{X,A} P(A|X,Y)}\right)\nonumber\\
&=-H(Y)+\log\mathbb{E}_P P(A|X,Y)-\mathbb{E}_{Y}\log\mathbb{E}_{X,A} P(A|X,Y).\label{Dy}
\end{align}
In addition,
\begin{align}
\log\mathbb{E}_Q \frac{1}{P(A|X,Y)}&=\log\sum_{x,y,a} \frac{Q(x,y,a)}{P(a|x,y)}\nonumber\\
&=-\log\mathbb{E}_P P(A|X,Y).\label{EQP}
\end{align}
Substituting~\eqref{Dy} and~\eqref{EQP} into~\eqref{LaplaceApp} yields
\begin{align}
H(A|X,Y)&=-H(Y)-\mathbb{E}_{Y}\log\mathbb{E}_{X,A} P(A|X,Y)+D\left(P_{X,A|Y}\|Q_{X,A|Y} \mid P_Y\right).\label{entropy2}
\end{align}
We are left with the task of evaluating the conditional relative entropy in~\eqref{entropy2}. The conditional distributions that participate in this term are given by
\begin{align}
P(x,a|y)&=P(x)P(a)\frac{\ind}{E_{X,A}\Indy}\label{Pxd}\\
Q(x,a|y)&=P(x)P(a)\frac{P(a|x,y)}{E_{X,A}P(A|X,y)}\label{Qxd}
\end{align}
and therefore
\begin{align}
D\left(P_{X,A|Y}\|Q_{X,A|Y} \mid P_Y\right)=\mathbb{E}_P\log\left(\frac{\Ind}{E_{X,A}\Ind}\cdot\frac{E_{X,A}P(A|X,Y)}{P(A|X,Y) } \right).\label{condDivergence}
\end{align}
Unfortunately, an exact computation of~\eqref{condDivergence} involves the computation of $\mathbb{E}_P\log (1/P(A|X,Y))$, which is the exact technical difficulty we are trying to avoid. Instead, we lower bound~\eqref{condDivergence} using the convexity of relative entropy, i.e.,
\begin{align}
D\left(P_{X,A|Y}\|Q_{X,A|Y} \mid P_Y\right) \geq D\left(P_{X,A}\|\tilde{Q}_{X,A}\right),\label{marginalizedDivergence}
\end{align}
where
\begin{align}
\tilde{Q}(x,a)&=\sum_{y}P(y)Q(x,a|y)\nonumber\\
&=P(x,a)\mathbb{E}_{Y}\frac{P(a|x,Y)}{E_{X,A}P(A|X,Y)}.\label{Qtildey}
\end{align}
Note that other properties of relative entropy, such as the data-processing inequality or Pinsker's inequality, could potentially be useful for bounding~\eqref{condDivergence}. Combining~\eqref{marginalizedDivergence} and~\eqref{Qtildey} gives,
\begin{align}
D\left(P_{X,A|Y}\|Q_{X,A|Y} \mid P_Y\right)\geq -\mathbb{E}_{X,A}\log\mathbb{E}_{Y}\frac{P(A|X,Y)}{E_{X,A}P(A|X,Y)}.\label{divergenceBound}
\end{align}
Substituting~\eqref{divergenceBound} into~\eqref{entropy2} and using~\eqref{cleancapacity} yields the following.

\vspace{1mm}

\begin{theorem}
  Let  $(X,Y)\sim P_X\times P_{Y|X}$ be jointly distributed discrete r.vs. Let $A$ be any action consistent with $P_{Y|X}$, i.e., such that $Y=A(X)$. Then
\begin{align}
I(X; Y)\geq-H(A)-\mathbb{E}_{Y}\log\mathbb{E}_{X,A} P(A|X,Y)-\mathbb{E}_{X,A}\log\mathbb{E}_{Y}\frac{P(A|X,Y)}{E_{X,A}P(A|X,Y)}.\label{entropyGenBound}
\end{align}
\label{thm:genLB}
\end{theorem}

\section{A Bound via Adjacency Events}\label{sec:bound-via-adjacency}

An action $A$ is called \emph{uniform} if all actions in its support $\m{A}$ are equiprobable, i.e.,
\begin{align}
P(a)=\begin{cases}
\frac{1}{|\m{A}|} & a\in\m{A}\\
0 & a\notin\m{A}
\end{cases}.\nonumber
\end{align}
At this point, we restrict our attention to this class of actions, for which the bound in Theorem~\ref{thm:genLB} takes a particularly simpler form that depends only on the marginal distributions of $X$ and $Y$ and their joint support. We then show that any channel can be essentially characterized by a uniform action, which in turn proves Theorem~\ref{thm:main}.

For any $x\in\mathcal{X}$ and $y\in\mathcal{Y}$ let
\begin{align}
\Actset(x,y)\triangleq\left\{a \ : \ a(x)=y \right\}\label{DxyDef}
\end{align}
be the set of all possible actions in $\Actset$ that map the input $x$ to the output $y$. Denote the cardinality of this set by $\Nxy\triangleq|\Actset(x,y)|$.
\begin{proposition}
If $A$ is a uniform action, then $A$ conditioned on $X$ and $Y$ is uniformly distributed over the set $\Actset(X,Y)$.\footnote{Note that the converse is not generally true. As a counterexample, consider the BSC with the action $A(X)=X\oplus Z$ where $Z\sim\textrm{Bern}(p)$.}
\label{prop:condDxy}
\end{proposition}
\begin{proof}
\begin{align}
P(a|x,y)&=\frac{P(x,y|a)P(a)}{P(x,y)}\nonumber\\
&=\frac{P(y|x,a)P(a)}{P(y|x)}\nonumber\\
&=\frac{\ind P(a)}{\sum_{a\in\Actset(x,y)}P(a)}\nonumber\\
&\stackrel{(a)}{=}\frac{\tfrac{1}{|\m{A}|}\I\left(a\in\m{A}(x,y)\right)}{\tfrac{1}{|\m{A}|}\Nxy}\nonumber\\
&=\frac{\I\left(a\in\m{A}(x,y)\right)}{\Nxy},\label{condDxy}
\end{align}
where $(a)$ follows from $\ind=\indinD$ and since $P(a)=\tfrac{1}{|\m{A}|}$ for all $a\in\m{A}$.
\end{proof}

\vspace{1mm}

\begin{lemma}
Suppose $P_{Y|X}$ can be represented by a uniform action $A$. Then, for any input distribution $P_X$
\begin{align}
-\mathbb{E}_{Y}\log\mathbb{E}_{X,A} P(A|X,Y)&=H(A)-\mathbb{E}_{Y}\log\mathbb{E}_{X}\Indnonempty.\label{Eylog}
\end{align}
\label{lem:FirstTermUnif}
\end{lemma}

\begin{proof}
Using Proposition~\ref{prop:condDxy},
\begin{align}
\mathbb{E}_{X,A} P(A|X,y)&=\sum_{x}P(x)\sum_{a}P(a)\frac{\I\left(a\in\m{A}(x,y)\right)}{\Nxy}\nonumber\\
&=\frac{1}{|\m{A}|}\sum_{x}P(x)\frac{\Nxy}{\Nxy}\indnonempty\nonumber\\
&=\frac{1}{|\m{A}|}\mathbb{E}_{X}\indnonemptyy.\label{ExdCond}
\end{align}
Thus,
\begin{align}
-\mathbb{E}_{Y}\log\mathbb{E}_{X,A} P(A|X,Y)&=-\mathbb{E}_{Y}\log \frac{1}{|\m{A}|}-\mathbb{E}_{Y}\log\mathbb{E}_{X}\Indnonempty\nonumber\\
&=\log{|\m{A}|}-\mathbb{E}_{Y}\log\mathbb{E}_{X}\Indnonempty.\nonumber
\end{align}
The lemma follows since $H(A)=\log{|\m{A}|}$ by the assumption that $A$ is a uniform action.
\end{proof}

\vspace{1mm}

The next lemma lower bounds the last term in~\eqref{entropyGenBound} for channels with a uniform action $A$.

\begin{lemma}
Suppose $P_{Y|X}$ can be represented by a uniform action $A$. Then, for any input distribution $P_X$
\begin{align}
-\mathbb{E}_{X,A}&\log\mathbb{E}_{Y}\frac{P(A|X,Y)}{E_{X,A}P(A|X,Y)}\geq-\mathbb{E}_{X}\log\mathbb{E}_{Y}\frac{\Indnonempty}{\mathbb{E}_{X}\Indnonempty}\geq 0\label{secondTermUnif}
\end{align}
\label{lem:SecondTermUnif}
\end{lemma}

\begin{proof}
By virtue of Jensen's inequality,
\begin{align}
-\mathbb{E}_{X,A}\log\mathbb{E}_{Y}&\frac{P(A|X,Y)}{E_{X,A}P(A|X,Y)}\geq -\mathbb{E}_{X}\log\mathbb{E}_{Y}\frac{\mathbb{E}_{A}P(A|X,Y)}{E_{X,A}P(A|X,Y)}.\nonumber
\end{align}
Using~\eqref{condDxy} and~\eqref{ExdCond}, we have
\begin{align}
\frac{\mathbb{E}_{A}P(A|x,y)}{E_{X,A}P(A|X,y)}&=\frac{\sum_{a}P(a)\frac{\I\left(a\in\m{A}(x,y)\right)}{\Nxy}}{\tfrac{1}{|\m{A}|}\mathbb{E}_{X}\indnonemptyy}\nonumber\\
&=\frac{\indnonempty}{\mathbb{E}_{X}\indnonemptyy},\nonumber
\end{align}
establishing the first inequality in~\eqref{secondTermUnif}. The second inequality follows by applying Jensen's inequality again, this time w.r.t. $\mathbb{E}_{X}$.
\end{proof}

\vspace{1mm}

Combining Theorem~\ref{thm:genLB}, Lemma~\ref{lem:FirstTermUnif}, and Lemma~\ref{lem:SecondTermUnif}, establishes the following.

\vspace{1mm}

\begin{lemma}
Suppose $P_{Y|X}$ can be represented by a uniform action $A$. Then, for any input distribution $P_X$
\begin{align}
I(X;Y)\geq&-\mathbb{E}_{Y}\log\mathbb{E}_{X}\Indnonempty-\mathbb{E}_{X}\log\mathbb{E}_{Y}\frac{\Indnonempty}{\mathbb{E}_{X}\Indnonempty}.
\label{capaityUnifBound}
\end{align}
\label{thm:uniformLB}
\end{lemma}

\vspace{1mm}

To establish our main result for any channel and input distribution, we first show the following.
\begin{lemma}
Let $P_{Y|X}$ be a channel with the property that $P(y|x)$ is rational for all $x$ and $y$. Then there exists a uniform action for $P_{Y|X}$.
\label{lem:rational}
\end{lemma}

\begin{proof}
For any channel $P_{Y|X}$ with rational probabilities there exists some action set $\m{A}=\{a_1,\cdots,a_{|\m{A}|}\}$ and a corresponding probability distribution $P_A$ consistent with it such that all probabilities $P_A(a_i)$, $i=1,\ldots,|\m{A}|$, are positive rational numbers. For example, the construction from Example~\ref{ex:genericaction} yields rational probabilities $P_A(a_i)$, $i=1,\ldots,|\m{A}|$. We construct a new action $\bar{A}$ by duplicating each action $a_i$ to $M_i$ identical actions, and assigning the probability $P_A(a_i)/M_i$ to each of them. Clearly, the new action is also consistent with $P_{Y|X}$ for any choice of the natural numbers $M_1,\ldots,M_{|\m{A}|}$. By our assumption that all original action probabilities are positive rational numbers, we can always find a choice of $M_1,\ldots,M_{|\m{A}|}$ such that all new action probabilities are equal. For such a choice the action $\bar{{A}}$ will be uniform.
\end{proof}

\vspace{1mm}

Using Lemma~\ref{thm:uniformLB} and Lemma~\ref{lem:rational}, we can now prove our main result.

\begin{proof}[Proof of Theorem~\ref{thm:main}]
Any channel $P_{Y|X}$ can be approximated arbitrarily well by a conditional distribution $\tilde{P}_{\tilde{Y}|X}$ with the same support whose entries are all rational, in the sense that $\max_{x,y}|P_{Y|X}(y|x) - \tilde{P}_{Y|X}(y|x)|$ can be made arbitrarily small. This means that both $P_X\times \tilde{P}_{Y|X}$ and the corresponding marginal $\tilde{P}_Y$  are arbitrarily close to $P_{XY}$ and $P_Y$ respectively. Since the mutual information $I(X;Y)$ is continuous with respect to $P_{XY}$, the mutual information $I(X;\tilde{Y})$ between $X$ and the output of the ``rational'' channel $\tilde{P}_{Y|X}$ can be made arbitrarily close to $I(X;Y)$. By Lemma~\ref{lem:rational}, there exists a uniform action for $\tilde{P}_{\tilde{Y}|X}$, and consequently by Lemma~\ref{thm:uniformLB} its mutual information is lower bounded by~\eqref{capaityUnifBound}. By continuity, $I(X;Y)$ is also lower bounded by~\eqref{capaityUnifBound}.
\end{proof}

\vspace{1mm}

\section{A Convex--Optimization Based Bound}\label{sec:conv}
In the previous section we have proved a lower bound on $I(X;Y)$ that depends only on the marginal distributions $P_X,P_Y$ and the support of the joint distribution, namely, the function $\indnonempty$. Our proof relied on information theoretic arguments. In this section we will take a more direct approach to the problem, and derive bounds on $I(X;Y)$ in terms of the same quantities, using convex optimization. More specifically, to arrive at a lower bound we minimize $I(X;Y)$ w.r.t. $P_{XY}$ subject to the constraints that the marginal distributions are $P_X,P_Y$, and that $P_{XY}(x,y)=0$ whenever $\indnonempty=0$. Throughout this section we assume all logarithms are in the natural basis, while the result of Theorem \ref{thm:convex_opt} remains valid as long as the same logarithmic basis is applied to $I(X;Y)$.

We consider the following problem:
\begin{align*}
&\text{minimize } I(X;Y) \\
&\text{subject to:}\\
&\sum_{x:x\sim y} P_{XY}(x,y)=P_Y(y) \qquad \forall y\in\m{Y}\\
&\sum_{y:y\sim x} P_{XY}(x,y)=P_X(x) \qquad \forall x\in\m{X}\\
&P_{XY}(x,y)\geq 0 \ \text{if } x\sim y, \qquad P_{XY}(x,y)= 0 \ \text{if } x\nsim y.
\end{align*}
Note that the constraints above imply $\sum_{x,y}P_{XY}(x,y)=1$. This is equivalent to
\begin{align*}
&\text{minimize } \sum_{x\sim y} P_{XY}(x,y)\log\frac{P_{XY}(x,y)}{P_X(x)P_Y(y)} \\
&\text{subject to:}\\
&\sum_{x:x\sim y} P_{XY}(x,y)=P_Y(y) \qquad \forall y\in\m{Y}\\
&\sum_{y:y\sim x} P_{XY}(x,y)=P_X(x) \qquad \forall x\in\m{X}\\
&P_{X,Y}(x,y)\geq 0 \qquad \forall (x,y)\in\m{X}\times\m{Y}, x\sim y
\end{align*}
This objective function is convex in $P_{XY}(x,y)$, and the constraints are linear, so the optimization solution can be obtained by the solution to the dual problem given by
\begin{align}
L&=\inf_{P_{XY}(x,y)}\sup_{\lambda_x,\mu_y\in\RR,\tau_{xy}\geq 0}\sum_{x\sim y}P_{XY}(x,y)\log\frac{P_{XY}(x,y)}{P_X(x)P_{Y}(y)}-\sum_{x}\lambda_x\left(\sum_{y: y\sim x}P_{XY}(x,y)-P_{X}(x) \right)\nonumber\\
& \ \ \ \ \ \ \ \ \ \ \ \ \ \ \ \ \ \ \ \ \ \ \ \ \ -\sum_{y}\mu_y\left(\sum_{x: x\sim y}P_{XY}(x,y)-P_{Y}(y) \right)-\sum_{x\sim y}\tau_{xy} P_{XY}(x,y)\nonumber\\
&=\sup_{\lambda_x,\mu_y\in\RR,\tau_{xy}\geq 0}\inf_{P_{XY}(x,y)}\sum_{x\sim y}P_{XY}(x,y)\left(\log\frac{P_{XY}(x,y)}{P_{X}(x)P_{Y}(y)}-\lambda_x-\mu_y-\tau_{xy}\right)+\sum_{x}\lambda_x P_{X}(x)+\sum_{y}\mu_y P_{Y}(y)\label{eq:minimax}\\
&=\sup_{\lambda_x,\mu_y\in\RR,\tau_{xy}\geq 0}\inf_{P_{XY}(x,y)}\sum_{x\sim y}P_{XY}(x,y)\left(\log P_{XY}(x,y)-a_{xy}\right)+\sum_{x}\lambda_x P_{X}(x)+\sum_{y}\mu_y P_{Y}(y)\nonumber
\end{align}
where~\eqref{eq:minimax} follows from the minimax theorem and
\begin{align}
a_{xy}\triangleq\log(P_{X}(x)P_{Y}(y))+\lambda_x+\mu_y+\tau_{xy}.\nonumber
\end{align}
The function $f(x)=x\log x-ax$ is minimized at $x^*=e^{a-1}$ and its minimal value is $f(x^*)=-e^{a-1}$. Using this, we get that
\begin{align}
L&=\sup_{\lambda_x,\mu_y\in\RR,\tau_{xy}\geq 0}\sum_{x\sim y}-e^{\log(P_{X}(x)P_{Y}(y))+\lambda_x+\mu_y+\tau_{xy}-1}+\sum_{x}\lambda_x P_{X}(x)+\sum_{y}\mu_y P_{Y}(y)\nonumber\\
&=\sup_{\lambda_x,\mu_y\in\RR,\tau_{xy}\geq 0}\sum_{x\sim y}-P_{X}(x)P_{Y}(y)e^{\lambda_x+\mu_y+\tau_{xy}-1}+\sum_{x}\lambda_x P_{X}(x)+\sum_{y}\mu_y P_{Y}(y)\nonumber
\end{align}
Clearly, the maximizing $\tau_{xy}$ is $\tau_{xy}=0$ which gives
\begin{align}
L&=\sup_{\lambda_x,\mu_y\in\RR}\sum_{x\sim y}-P_{X}(x)P_{Y}(y)e^{\lambda_x+\mu_y-1}+\sum_{x}\lambda_x P_{X}(x)+\sum_{y}\mu_y P_{Y}(y)\nonumber\\
&=\sup_{\lambda_x,\mu_y\in\RR}\sum_{x\sim y}-P_{X}(x)P_{Y}(y)e^{\lambda_x+\mu_y}+\sum_{x}\lambda_x P_{X}(x)+\sum_{y}\mu_y P_{Y}(y)+1\nonumber
\end{align}
where in the last step we replaced $\mu_y$ with $\mu_y-1$ (with some abuse of notation).
Let $\bm{\lambda}$ and $\bm{\mu}$ be the vectors holding $\{\lambda_x\}_{x\in\m{X}}$ and $\{\mu_y\}_{y\in\m{Y}}$, respectively, and
\begin{align}
G(\bm{\lambda},\bm{\mu})\triangleq\sum_{x\sim y}-P_{X}(x)P_{Y}(y)e^{\lambda_x+\mu_y}+\sum_{x}\lambda_x P_{X}(x)+\sum_{y}\mu_y P_{Y}(y)+1,\nonumber
\end{align}
such that
\begin{align}
L=\sup_{\bm{\lambda}\in\RR^{|\m{X}|},\bm{\mu}\in\RR^{|\m{Y}|}} G(\bm{\lambda},\bm{\mu}).\nonumber
\end{align}

We will use the alternating minimization approach to minimize $-G(\bm{\lambda},\bm{\mu})$ (which is equivalent to maximizing $G(\bm{\lambda},\bm{\mu})$) over $\RR^{|\m{X}|}\times \RR^{|\m{Y}|})$. This approach is described as follows: for arbitrary initialization of $\bm{\lambda}^{(0)}$, we use an iterative algorithm to successively minimize the target function. In $k$-th iteration, we first hold $\bm{\lambda}^{(k-1)}$ fixed and minimize the target function over $\bm{\mu}$ to obtain $\bm{\mu}^{(k-1)}$, and then hold $\bm{\mu}^{(k-1)}$ fixed and minimize the the target function over $\bm{\lambda}$ to obtain $\bm{\lambda}^{(k)}$. In mathematical forms, for $k\ge 0$, we have
\begin{align*}
\bm{\mu}^{(k)} &\in \argmin_{\bm{\mu}} -G(\bm{\lambda}^{(k)},\bm{\mu}),\\
\bm{\lambda}^{(k+1)} &\in \argmin_{\bm{\lambda}} -G(\bm{\lambda},\bm{\mu}^{(k)}).
\end{align*}
The alternating minimization approach is widely used in optimization where separate optimization over different parameter subsets is much easier than the joint optimization, e.g., in the expectation--minimization (EM) algorithm \cite{dempster1977maximum} to find the maximum likelihood estimator, in the Blahut--Arimoto algorithm \cite{blahut1972computation,arimoto1972algorithm} to maximize the mutual information between channel input and output, in minimizing the Kullback--Leibler divergence between two convex sets of finite measures \cite{csisz1984information}, to name a few. One remarkable property of this approach is that, by definition we have
\begin{align}
G(\bm{\lambda}^{(0)},\bm{\mu}^{(0)}) \le G(\bm{\lambda}^{(1)},\bm{\mu}^{(0)}) \le G(\bm{\lambda}^{(1)},\bm{\mu}^{(1)}) \le G(\bm{\lambda}^{(2)},\bm{\mu}^{(1)}) \le \cdots \le L
\end{align}
i.e., the value sequence obtained by this approach is non-decreasing and must have a limit. We remark that since $G(\bm{\lambda},\bm{\mu})$ is jointly concave with respect to $(\bm{\lambda},\bm{\mu})$, the alternating minimization approach converges to the global optima \cite[Proposition 2.7.1]{bertsekas1999nonlinear}, i.e., $$\lim_{k\to\infty}G(\bm{\lambda}^{(k)},\bm{\mu}^{(k)})=\lim_{k\to\infty}G(\bm{\lambda}^{(k+1)},\bm{\mu}^{(k)}) = L.$$

Next, we derive the expression of $\bm{\lambda}^{(k)}$ and $\bm{\mu}^{(k)}$ obtained from the alternating minimization procedure. Initially we set $\lambda^{(0)}_x=0, \ \forall x\in\m{X}$. For $k\ge 0$, by the definition of $\bm{\mu}^{(k)}$ in the alternating minimization we have $\tfrac{\partial G}{\partial \mu_y^{(k)}}=0, \ \forall y\in\m{Y}$, which gives
\begin{align}
e^{-\mu_y^{(k)}}=&\sum_{x:x\sim y}P_{X}(x)e^{\lambda_x^{(k)}}, \qquad \forall y\in\m{Y}\label{eq:xeqs}.
\end{align}
Similarly, for $\bm{\lambda}^{(k+1)}$ we have
\begin{align}
e^{-\lambda_x^{(k+1)}}=&\sum_{y:y\sim x}P_{Y}(y)e^{\mu_y^{(k)}}, \qquad \forall x\in\m{X}\label{eq:yeqs}.
\end{align}

Based on (\ref{eq:yeqs}) and (\ref{eq:xeqs}), it is straightforward to verify that the first two iterations of this procedure yield
\begin{align}
L&\geq G\left(\bm{\lambda}^{(0)},\bm{\mu}^{(0)}\right)=-\mathbb{E}_Y\log\mathbb{E}_X \mathds{1}(X\sim Y)\nonumber\\
L&\geq G\left(\bm{\lambda}^{(1)},\bm{\mu}^{(0)}\right)=-\mathbb{E}_Y\log\mathbb{E}_X \mathds{1}(X\sim Y)-\mathbb{E}_X\log\mathbb{E}_Y\frac{\mathds{1}(X\sim Y)}{\mathbb{E}_X \mathds{1}(X\sim Y)}\nonumber
\end{align}
in agreement with the bound derived in Theorem~\ref{thm:main}.
Continuing with this procedure we can further improve our bound. To characterize the bound after $k$ iterations, we introduce the functions $T_X^{(k)}(P_{X}(x),P_{Y}(y),\indnonempty), T_Y^{(k)}(P_{X}(x),P_{Y}(y),\indnonempty)$ that are defined recursively as
\begin{align}
T_X^{(0)}(P_{X}(x),P_{Y}(y),\indnonempty)&=1,
\end{align}
and for $k\ge 0$,
\begin{align}
T_Y^{(k)}(P_{X}(x),P_{Y}(y),\indnonempty)&=\mathbb{E}_X\left(\frac{\mathds{1}(X\sim Y)}{T_X^{(k)}(P_{X}(x),P_{Y}(y),\indnonempty)}\right),\\
T_X^{(k+1)}(P_{X}(x),P_{Y}(y),\indnonempty)&=\mathbb{E}_Y\left(\frac{\mathds{1}(X\sim Y)}{T_Y^{(k)}(P_{X}(x),P_{Y}(y),\indnonempty)}\right).
\end{align}
It can be easily verified by induction that
\begin{align}
G\left(\bm{\lambda}^{(k)},\bm{\mu}^{(k)}\right)=&-\mathbb{E}_X\log T_X^{(k)}(P_{X}(x),P_{Y}(y),\indnonempty)
-\mathbb{E}_Y\log T_Y^{(k)}(P_{X}(x),P_{Y}(y),\indnonempty)\nonumber\\
G\left(\bm{\lambda}^{(k+1)},\bm{\mu}^{(k)}\right)=&-\mathbb{E}_X\log T_X^{(k+1)}(P_{X}(x),P_{Y}(y),\indnonempty)\nonumber-\mathbb{E}_Y\log T_Y^{(k)}(P_{X}(x),P_{Y}(y),\indnonempty)\nonumber.
\end{align}
Thus, we have arrived at the following theorem.
\begin{theorem}\label{thm:convex_opt}
For any jointly distributed discrete r.vs $(X,Y)$ and any $k\ge 0$,
\begin{align}
I(X;Y)\geq &-\mathbb{E}_X\log T_X^{(k+1)}(P_{X}(x),P_{Y}(y),\indnonempty)
-\mathbb{E}_Y\log T_Y^{(k)}(P_{X}(x),P_{Y}(y),\indnonempty)\nonumber\\
\geq &-\mathbb{E}_X\log T_X^{(k)}(P_{X}(x),P_{Y}(y),\indnonempty)\nonumber-\mathbb{E}_Y\log T_Y^{(k)}(P_{X}(x),P_{Y}(y),\indnonempty)\nonumber.
\end{align}
\end{theorem}

\section{Examples}\label{sec:examples}

In this section we evaluate the bounds derived in Theorems~\ref{thm:main} and \ref{thm:upper}, and when possible also those from Theorem~\ref{thm:convex_opt}, for four examples. The following simple lower bound on $I(X;Y)$ will serve as our baseline for demonstrating the improvement attained by applying the bound from Theorem~\ref{thm:main}.

\begin{proposition}
  For any jointly distributed discrete r.vs $(X,Y)$,
  \begin{align}\label{eq:basebound}
I(X,Y) \geq -\Expt_Y\log\Expt_X\I(X\adj Y).
  \end{align}
\label{prop:baseline}
\end{proposition}

Similar to the bound from Theorem~\ref{thm:main}, the bound above is given in terms of the marginals and the joint support of $(X,Y)$. However it is weaker than the former bound as it can be obtained from it directly by applying Jensen's inequality on the second term of~\eqref{eq:ourbound}, which gives  $-\mathbb{E}_{X}\log\mathbb{E}_{Y}\frac{\Indnonempty}{\mathbb{E}_{X}\Indnonempty}\geq -\log\mathbb{E}_{Y}\frac{\mathbb{E}_{X}\Indnonempty}{\mathbb{E}_{X}\Indnonempty}=0$. In Section~\ref{sec:intro} we also gave an operational proof of this bound.

\subsection{Erasure Channel}
The binary erasure channel has input $X\in\{0,1\}$ and output $Y\in\{0,1,\m{E}\}$ such that $\Pr(Y=x|X=x)=1-\epsilon$ and $\Pr(Y=\m{E}|X=x)=\epsilon$ for any $x$. For $X\sim\mathrm{Bern}(p)$ we have $\Pr(Y=0)=(1-\epsilon)(1-p)$, $\Pr(Y=1)=(1-\epsilon)p$ and $\Pr(Y=\m{E})=\epsilon$ and the mutual information between the input and output is $I_p(X;Y)=(1-\epsilon)h(p)$. For this channel $x\sim y$ if and only if either $x=y$ or $y=\m{E}$, and therefore
\begin{align}
-\Expt_Y\log\Expt_X\I(X\adj Y)&=-(1-\epsilon)(1-p)\log (1-p)-(1-\epsilon)p\log(p)-\epsilon\log(1)\nonumber\\
&=(1-\epsilon)h(p).\nonumber
\end{align}
Thus, for this channel our lower bound from Theorem~\ref{thm:main} as well as the weaker bound from Proposition~\ref{prop:baseline} are tight.

In order to evaluate our upper bound from Theorem~\ref{thm:upper} we need to choose an action $A$ consistent with $P_{Y|X}$. We take the natural action set, that consists of two actions, $a_1(x)=x$ and $a_2(x)=\m{E}$ with $p(a_1)=1-\epsilon$ and $p(a_2)=\epsilon$. For this choice we have
\begin{align}
\Expt_A\log \Expt_{XY}I(A\sim(X,Y))&=p(a_1)\log\Pr(X=Y)+p(a_2)\log\Pr(Y=\m{E})\nonumber\\
&=(1-\epsilon)\log(1-\epsilon)+\epsilon\log(\epsilon)\nonumber\\
&=-h(\epsilon).\nonumber
\end{align}
Since $H(Y)=h(\epsilon)+(1-\epsilon)h(p)$, the upper bound from Theorem~\ref{thm:upper} is tight and gives
\begin{align}
I_p(X;Y)\leq (1-\epsilon)h(p).\nonumber
\end{align}

\subsection{Z Channel}

The (symmetric) Z channel has a binary input $X$ and a binary output $Y$ such that $\Pr(Y=0|X=0)=1$ and $\Pr(Y=0|X=1)=\Pr(Y=1|X=1)=\tfrac{1}{2}$. For $X\sim\mathrm{Bern}(p)$ we have $Y\sim\mathrm{Bern}(\tfrac{p}{2})$, and the mutual information between the input and output is $I_p(X;Y)=h(\tfrac{p}{2})-p$. For this channel $x\sim y$ if and only if $(x,y)\neq (0,1)$ and therefore $\Expt_X\I(X\adj 0)=1$ and $\Expt_X\I(X\adj 1)=\Pr(X=1)=p$. We have
\begin{align}
-\Expt_Y\log\Expt_X\I(X\adj Y)&=-\Pr(Y=0)\log\Expt_X\I(X\adj 0)-\Pr(Y=1)\log\Expt_X\I(X\adj 1) \nonumber\\
&=-\frac{p}{2}\log(p),\label{Zfirstterm}
\end{align}
and
\begin{align}
-\mathbb{E}_{X}\log\mathbb{E}_{Y}\frac{\Indnonempty}{\mathbb{E}_{X}\Indnonempty}&=-(1-p)\log\left(\left(1-\frac{p}{2}\right)\frac{\I(0\sim 0)}{1}+\frac{p}{2}\frac{\I(0\sim 1)}{p} \right)
-p\log\left(\left(1-\frac{p}{2}\right)\frac{\I(1\sim 0)}{1}+\frac{p}{2}\frac{\I(1\sim 1)}{p} \right)\nonumber\\
&=-(1-p)\log\left(1-\frac{p}{2}\right)
-p\log\left(\left(1-\frac{p}{2}\right)+\frac{1}{2}\right)\nonumber\\
&=1-(1-p)\log(2-p)-p\log(3-p).
\label{ZsecondTerm}
\end{align}
Thus, Proposition~\ref{prop:baseline} gives
\begin{align}
I_p(X;Y)\geq -\frac{p}{2}\log(p),\label{Zsimple}
\end{align}
and Theorem~\ref{thm:main} gives
\begin{align}
I_p(X;Y)\geq 1-\frac{p}{2}\log(p)-(1-p)\log(2-p)-p\log(3-p).\label{Zub}
\end{align}

For comparison, we also take a look at a further refinement given by Theorem \ref{thm:convex_opt}. By the definitions of $T_X^{(k)}$ and $T_Y^{(k)}$, we know that
\begin{align*}
T_Y^{(0)}(P_{X}(x),P_{Y}(y),\indnonempty) &= \mathbb{E}_X \I(X\sim Y) = \I(Y=0) + p\I(Y=1),\\
T_X^{(1)}(P_{X}(x),P_{Y}(y),\indnonempty) &= \mathbb{E}_Y\left(\frac{\mathds{1}(X\sim Y)}{T_Y^{(0)}(P_{X}(x),P_{Y}(y),\indnonempty)}\right) = \mathbb{E}_Y\left(\frac{\mathds{1}(X\sim Y)}{\I(Y=0) + p\I(Y=1)}\right)\\
&= \left(1-\frac{p}{2}+\frac{p}{2}\cdot 0\right)\I(X=0) + \left(1-\frac{p}{2}+\frac{p}{2}\cdot \frac{1}{p}\right)\I(X=1)\\
&= \left(1-\frac{p}{2}\right)\I(X=0) + \left(\frac{3}{2}-\frac{p}{2}\right)\I(X=1),\\
T_Y^{(1)}(P_{X}(x),P_{Y}(y),\indnonempty) &= \mathbb{E}_X\left(\frac{\mathds{1}(X\sim Y)}{T_X^{(1)}(P_{X}(x),P_{Y}(y),\indnonempty)}\right) =\mathbb{E}_X\left(\frac{\mathds{1}(X\sim Y)}{\left(1-\frac{p}{2}\right)\I(X=0) + \left(\frac{3}{2}-\frac{p}{2}\right)\I(X=1)}\right)\\
&= \left((1-p)\cdot \frac{1}{1-\frac{p}{2}} + p\cdot \frac{1}{\frac{3}{2}-\frac{p}{2}}\right)\I(Y=0) + \left((1-p)\cdot 0 + p\cdot \frac{1}{\frac{3}{2}-\frac{p}{2}} \right)\I(Y=1)\\
&= \left(\frac{2-2p}{2-p} + \frac{2p}{3-p}\right)\I(Y=0) + \frac{2p}{3-p}\I(Y=1).
\end{align*}
As a result, Theorem \ref{thm:convex_opt} gives
\begin{align}\label{Zconv}
I_p(X;Y) &\geq -\mathbb{E}_X\log T_X^{(1)}(P_{X}(x),P_{Y}(y),\indnonempty)\nonumber-\mathbb{E}_Y\log T_Y^{(1)}(P_{X}(x),P_{Y}(y),\indnonempty)\nonumber\\
&= -(1-p)\log\left(1-\frac{p}{2}\right) - p\log\left(\frac{3-p}{2}\right) - \left(1-\frac{p}{2}\right)\log\left(\frac{2-2p}{2-p} + \frac{2p}{3-p}\right) - \frac{p}{2}\log\left(\frac{2p}{3-p}\right)\nonumber\\
&= \frac{p}{2}\log(2-p)+(1-p)\log(3-p)-\frac{p}{2}\log(p)-\left(1-\frac{p}{2}\right)\log(3-2p).
\end{align}

The bounds from~\eqref{Zsimple},~\eqref{Zub} and~\eqref{Zconv} are plotted in Figure~\ref{fig:Z} as a function of $p$ along with the exact value of $I_p(X;Y)$. It can be seen that the lower bound from Theorem~\ref{thm:main} is significantly tighter than the one form Proposition~\ref{prop:baseline}, and it is quite close to $I_p(X;Y)$ for all values of $p$. The lower bound from Theorem \ref{thm:convex_opt} is even tighter.

\begin{figure}
\begin{center}
\includegraphics[width=0.8 \columnwidth]{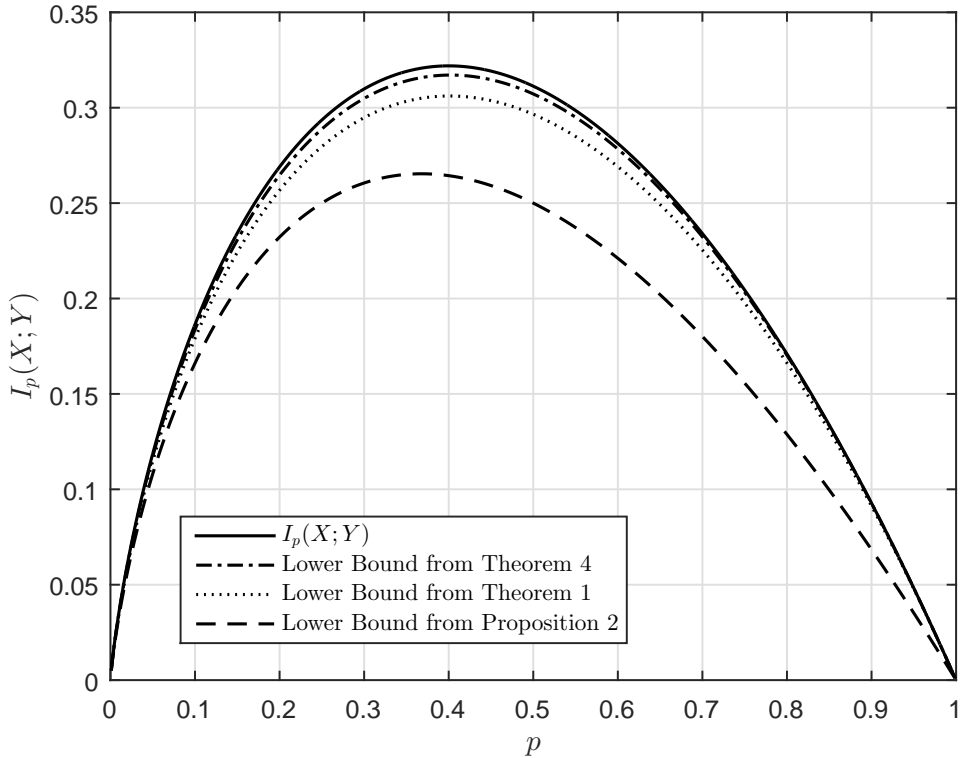}
\end{center}
\caption{$I_p(X;Y)$ for the Z channel together with the three lower bounds from~\eqref{Zsimple},~\eqref{Zub} and~\eqref{Zconv} as a function of $p$.}
\label{fig:Z}
\end{figure}

In order to evaluate the upper bound from Theorem~\ref{thm:upper} we use the natural action $a_1(x)=x$ and $a_2(x)=0$ with $p(a_1)=p(a_2)=\tfrac{1}{2}$. For this choice $a_1\sim(x,y)$ if and only if $x=y$ and $a_2\sim(x,y)$ if and only if $y=0$, and therefore $\Expt_{XY}\I(a_1\sim(X,Y))=\Pr(X=Y)=1-\tfrac{p}{2}$ and $\Expt_{XY}\I(a_2\sim(X,Y))=\Pr(Y=0)=1-\tfrac{p}{2}$. We have
\begin{align}
\Expt_A\log \Expt_{XY}\I(A\sim(X,Y))&=\log\left(1-\frac{p}{2}\right),\label{ZubPart1}
\end{align}
and
\begin{align}
&\Expt_{X,Y}\log\Expt_A\frac{\I(A\sim (X,Y))}{\Expt_{X,Y}\I(A\sim (X,Y))}=\Expt_{X,Y}\log\Expt_A\I(A\sim (X,Y))-\log\left(1-\frac{p}{2}\right)\nonumber\\
&=\Pr(X=0,Y=0)\log(1)+\Pr(X=1,Y=0)\log\left(\frac{1}{2}\right)+\Pr(X=1,Y=1)\log\left(\frac{1}{2}\right)-\log\left(1-\frac{p}{2}\right)\nonumber\\
&=-p-\log\left(1-\frac{p}{2}\right).
\label{ZubPart2}
\end{align}
Recalling that $H(Y)=h(\tfrac{p}{2})$ and applying theorem~\ref{thm:upper} we get
\begin{align}
I_p(X;Y)&\leq h\left(\frac{p}{2}\right)+\log\left(1-\frac{p}{2}\right)-p-\log\left(1-\frac{p}{2}\right)\nonumber\\
&=h\left(\frac{p}{2}\right)-p\nonumber,
\end{align}
which is tight for any $p$.

\subsection{Binary Deletion Channel}
The binary i.i.d. deletion channel operates by independently deleting input bits with probability $d$. In this subsection, we apply Theorem~\ref{thm:main}  and Theorem~\ref{thm:upper} to obtain lower and upper bounds on the mutual information for an i.i.d. uniform input process. Both bounds outperform the best known bounds in some regimes of deletion probabilities. In general, tighter lower bounds can be obtained by applying Theorem~\ref{thm:convex_opt} with higher values of $k$. However, as will be demonstrated below, even the task of computing the bound from Theorem~\ref{thm:main} (corresponding to Theorem~\ref{thm:convex_opt} with $k=0$) is quite challenging.

\vspace{3pt}

\noindent\textit{\underline{Lower Bound for an  i.i.d Uniform Input}}

We apply Theorem~\ref{thm:main} to obtain a lower bound for $I(\bX;\bY)$ under a uniform i.i.d. input distribution $\bX\sim\Unif \left(\{0,1\}^n\right)$, which outperforms the best known bounds for i.i.d inputs \cite{gallager1961,rd13}. Since the deletion channel is information stable, any rate smaller than the associated $\lim_{n\rightarrow\infty}I(\bX;\bY)/n$ is achievable with uniform i.i.d. codebooks. Note that for a uniform i.i.d. input, the output $\bY$ is also uniform i.i.d. given its length $\Theta n$, where the latter is binomial with parameters $(n,1-d)$.

For the i.i.d. deletion channel 
$\mathds{1}(\bx\sim\by)$ indicates whether or not $\by$ is a subsequence of $\bx$. For $0\leq t\leq 1$, define the operation $\langle t\rangle\triangleq\max(t,1/2)$.
According to~\cite[Lemma 3.1]{dg01}, for any $\by$ of length $\theta n$ we have
\begin{align}
\sum_{\bx\in\{0,1\}^n}\mathds{1}(\bx\sim\by)=\sum_{j=\theta n}^n {n\choose j}\doteq 2^{n h(\langle \theta\rangle)},\label{diggavicount}
\end{align}
where $h(\cdot)$ is the binary entropy function, and $\doteq$ denotes exponential equality in the usual sense. This implies that for any $\by$ of length $\theta n$ we have $\mathbb{E}_{\bX}\mathds{1}(\bX\sim\by)\doteq 2^{n (h(\langle \theta\rangle)-1)}$. The function $h(\langle \theta\rangle)$ is concave in $\theta$ and therefore
\begin{align}
-\lim_{n\to\infty}\frac{1}{n}\mathbb{E}_{Y}\log\mathbb{E}_{X}\mathds{1}(\bX\sim\bY)&= -\mathbb{E}_{\Theta}\left(h(\langle \Theta\rangle)-1\right)\nonumber\\
&\geq 1-h(\langle \mathbb{E}\Theta \rangle)\nonumber\\
&= 1-h(\langle 1-d \rangle).\label{deletion1stTerm}
\end{align}
where $\Theta$ is the normalized (random) length of $Y$.

The right hand side of~\eqref{deletion1stTerm} is a well known lower bound for the deletion channel capacity, obtained with a uniform i.i.d. input~\cite{dg01}. We now evaluate the second term in~\eqref{eq:ourbound} in order to improve upon this bound. To this end, we first parse each $x\in\{0,1\}^n$ into phrases that contain exactly two bit flips and end immediately after the second flip. For example, the string $0001111011001110001$ is parsed into the three phrases $00011110,11001,110001$. We identify each phrase with three parameters: $b\in\{0,1\}$ is the first bit in the phrase, $k_1\geq 2$ is the index of the first flip in the phrase, and $k_2\geq 1$ is such that $k_1+k_2$ is the total number of bits in the phrase. In our example, the three phrases correspond to $\{b=0,k_1=4,k_2=4\}$, $\{b=1,k_1=3,k_2=2\}$ and $\{b=1,k_1=3,k_2=3\}$, respectively. For any pair of integers $2\leq k_1<n$, $1\leq k_2<n$ let $\Psi^{k_1,k_2}(\bx)$ be the number of $\{k_1,k_2\}$-phrases in the parsing of $\bx$. For $\epsilon>0$ we define the typical set
\begin{align}
\Tx_{\epsilon}\triangleq\Big\{\bx\in\{0,1\}^n \ : & \ \left|\frac{1}{n}\Psi^{k_1,k_2}(\bx)-\frac{1}{5}\cdot 2^{-(k_1+k_2-1)}\right|<\epsilon \ \ \ \ \ \ \forall \ 2\leq k_1<n, \  1\leq k_2<n  \Big\}\nonumber.
\end{align}
It holds that for any $\epsilon>0$ and $n$ large enough $\Pr(\bX\in\Tx_{\epsilon})$ is indeed arbitrary close to $1$. To see this, define the three i.i.d. mutually independent processes
\begin{align}
B_i&\sim\mathrm{Bern}(\tfrac{1}{2}), \ \text{i.i.d.}\nonumber\\
K_{1i}&\sim 1+\mathrm{Geometric}(\tfrac{1}{2}), \ \text{i.i.d.}\nonumber\\
K_{2i}&\sim \mathrm{Geometric}(\tfrac{1}{2}), \ \text{i.i.d.}\nonumber
\end{align}
and note that an i.i.d. $\mathrm{Bern}(\tfrac{1}{2})$ random process is equivalent to the process obtained by stacking the random phrases $\{B_i,K_{1i},K_{2i}\}$ one after the other. Moreover, the probability of such a random phrase being of type $\{k_1,k_2\}$ is $2^{-(k_1+k_2-1)}$ and the expected length is $\mathbb{E}(K_{1i}+K_{2i})=5$. In our setting, $\bX$ is an $n$-dimensional i.i.d. $\mathrm{Bern}(\tfrac{1}{2})$ random vector. Thus, $\bX$ can be generated by stacking exactly $n/5$ random phrases $\{B_i,K_{1i},K_{2i}\}$ one after the other and either removing the last bits if the length of the obtained vector is greater than $n$, or appending i.i.d. $\mathrm{Bern}(\tfrac{1}{2})$ bits to the vector if its length is smaller than $n$. Since the expected length of a phrase is $5$ bits, for any $\delta>0$ the number of removed/appended bits is w.h.p. smaller than $\delta n$. Therefore, the contribution of these bits to the distribution of the phrase lengths in the parsing of $\bX$ is negligible, and we get that $\Pr(\bX\in\Tx_{\epsilon})\to 1$ with $n$, by the law of large numbers.

For $n$ large enough we can write
\begin{align}
-\mathbb{E}_{\bX}\log\mathbb{E}_{\bY}\frac{\mathds{1}(\bX\sim\bY)}{\mathbb{E}_{\bX}\mathds{1}(\bX\sim\bY)}&=
-\Pr(\bX\in\Tx_{\epsilon})\mathbb{E}_{\bX|\Tx_\epsilon}\log\mathbb{E}_{\bY}\frac{\mathds{1}(\bX\sim\bY)}{\mathbb{E}_{\bX}\mathds{1}(\bX\sim\bY)}
-\Pr(\bX\notin\Tx_{\epsilon})\mathbb{E}_{\bX|\overline{\Tx}_\epsilon}\log\mathbb{E}_{\bY}\frac{\mathds{1}(\bX\sim\bY)}{\mathbb{E}_{\bX}\mathds{1}(\bX\sim\bY)}\nonumber\\
&\geq -\Pr(\bX\in\Tx_{\epsilon})\log\mathbb{E}_{\bY}\frac{\mathbb{E}_{\bX|\Tx_\epsilon}\mathds{1}(\bX\sim\bY)}{\mathbb{E}_{\bX}\mathds{1}(\bX\sim\bY)}
-\Pr(\bX\notin\Tx_{\epsilon})\log\mathbb{E}_{\bY}\frac{\mathbb{E}_{\bX|\overline{\Tx}_\epsilon}\mathds{1}(\bX\sim\bY)}{\mathbb{E}_{\bX}\mathds{1}(\bX\sim\bY)}\nonumber\\
&\geq -(1-\epsilon)\log\mathbb{E}_{\bY}\frac{\mathbb{E}_{\bX|\Tx_\epsilon}\mathds{1}(\bX\sim\bY)}{\mathbb{E}_{\bX}\mathds{1}(\bX\sim\bY)}-\epsilon n,\label{condS}
\end{align}
where the first inequality follows from Jensen's inequality and in the second we have used the fact that $\mathbb{E}_{X}\mathds{1}(\bX\sim\by)\geq 2^{-n}$ for any $\by$, and therefore $\mathds{1}(\bX\sim\bY)/\mathbb{E}_{X}\mathds{1}(\bX\sim\by)\leq 2^n$ for any $\by$, along with  $\Pr(\bX\in\Tx_{\epsilon})>1-\epsilon$.
Recalling that $\Theta$ is the normalized (random) length of $\bY$, we take the expectation $\mathbb{E}_{\bY}$ as $\mathbb{E}_{\Theta}\mathbb{E}_{\bY|\Theta}$ and use~\eqref{diggavicount} to obtain
\begin{align}
\mathbb{E}_{\bY}\frac{\mathbb{E}_{\bX|\Tx_{\epsilon}}\mathds{1}(\bX\sim\bY)}{\mathbb{E}_{\bX}\mathds{1}(\bX\sim\bY)}&\doteq\mathbb{E}_{\Theta}2^{n (1-h(\langle \Theta \rangle))}\mathbb{E}_{\bY|\Theta}\mathbb{E}_{\bX|\Tx_\epsilon}\mathds{1}(\bX\sim\bY)\nonumber\\
&=\mathbb{E}_{\Theta}2^{n (1-h(\langle \Theta \rangle))}\Pr\left(\bX\sim \bY|\Theta,\bX\in\Tx_\epsilon \right).\label{expectedPerror}
\end{align}

Now, consider a greedy algorithm for determining whether $\by$ is a subsequence of $\bx$, defined as follows~\cite[Section 3.1]{mitzenmacher09}: Scanning from left to right, take the first bit in $\by$ and match it with its first appearance in $\bx$. Then take the second bit in $\by$ and match it with its subsequent first appearance in $\bx$. Continue until either $\bx$ or $\by$ are exhausted, where the latter case is termed success. It is easy to see that the greedy algorithm succeeds if and only if $\bx\sim \by$. 
For statistically independent random vectors $\bX$ and $\bY$, we enumerate the phrases in the parsing of $\bX$ by $i=1,\ldots,M(\bX)$ where $M(\bX)$ is the (random) number of phrases in $\bX$. The vector $\bY$ consists of $\Theta n$ i.i.d. uniform bits. To simplify computations, we construct a vector $\bY'$ of length $n$ by taking $\bY$ and possibly padding it with i.i.d. bits. We define the random variables $Z_i$ as the number of bits in $\bY'$ that are matched to bits in the $i$th phrase of $\bX$ by the greedy algorithm. Under this construction, the events $\left\{\sum_i Z_i\geq \Theta n\right\}$ coincides with the event $\{\bX\sim \bY\}$, since the additional random suffix does not affect the event where the first $\Theta n$ bits in $\bY'$ are matched. Under this assumption the $Z_i$'s are clearly mutually independent, given that the phrase types $\{k_{1i},k_{2i}\}_{i=1}^{M(\bX)}$ of $\bX$ are known (but assuming that their first bit identifiers $\{b_i\}_{i=1}^{M(\bX)}$ remain random). Of course, the distribution of $Z_i$ depends on the parameters $k_{1i},k_{2i}$ that correspond to the $i$th phrase in $\bX$. In the appendix, we show that given $K_{1i}$ and $K_{2i}$, the (base two) moment generating function of $Z_i$ is
\begin{align}
\lambda^{k_1,k_2}_{Z_{i}}(t)&\triangleq \mathbb{E}\left(2^{t Z_i}|{K_{1i}=k_1,K_{2i}=k_2}\right)\nonumber\\
&=2^{k_1(t-1)}+2^{t-1}\frac{1-2^{k_1(t-1)}}{1-2^{t-1}}\left(2^{t-1}\frac{1-2^{k_2(t-1)}}{1-2^{t-1}}+2^{k_2(t-1)-t} \right).\nonumber
\end{align}
Noting that by definition, for $\bX\in\Tx_\epsilon$ the number of phrases $M(\bX)$ and their composition $\Psi^{k_1,k_2}(\bX)$ is essentially deterministic, we can use Chernoff's bound~\cite{dembozeitouni} to obtain
\begin{align*}
\Pr(\bX\sim \bY|\Theta = \theta,\bX\in\Tx_\epsilon)&=\Pr\left(\sum_{i=1}^{M(\bX)}Z_i\geq \Theta n|\Theta = \theta,\bX\in\Tx_\epsilon\right)\nonumber\\
&\dotleq 2^{-n\Lambda^*(\theta)},
\end{align*}
where
\begin{align*}
\Lambda^*(\theta)=\max_{t>0}\left(\theta t-\frac{1}{5}\sum_{k_1=2}^\infty\sum_{k_2=1}^\infty2^{-(k_1+k_2-1)}\log \lambda_{Z_{k_1,k_2}}(t)\right).
\end{align*}
Substituting into~\eqref{condS} and~\eqref{expectedPerror}, and applying standard large deviations arguments, we obtain
\begin{align*}
-\lim_{n\to\infty}\frac{1}{n}\mathbb{E}_{\bX}\log\mathbb{E}_{\bY}\frac{\mathds{1}(\bX\sim\bY)}{\mathbb{E}_{\bX}\mathds{1}(\bX\sim\bY)}\geq g(d)
\end{align*}
where
\begin{align*}
g(d)\triangleq \min_{0\leq \theta\leq 1} D_2(\theta\|1-d)-(1-h(\langle\theta\rangle))+\Lambda^*(\theta)
\end{align*}
where $D_2(p\|q)$ is the binary relative entropy function. It follows that for a uniform i.i.d. input distribution,
\begin{align}
\lim_{n\rightarrow\infty}\frac{1}{n}I(\bX;\bY)\geq 1- h(\min (d,1/2))+g(d).\label{deletionBound}
\end{align}
Numerical evaluation of the term $g(d)$ reveals that it is greater than zero for all $d<1/2$. Thus,~\eqref{deletionBound} improves over Gallager's well know bound $1-h(d)$ \cite{gallager1961}. Recently, Rahmati and Duman~\cite{rd13} used a different technique to lower bound the mutual information for uniform i.i.d. inputs. For small values of $d$ their bound is better than~\eqref{deletionBound}, but for larger values of $d$ the right hand side of~\eqref{deletionBound} turns out to be greater than their bound. For example, for $d=0.2$ our bound improves on $1-h(0.2)$ by $\approx 0.0117$ bits (roughly $5\%$), whereas the improvement of~\cite{rd13} is negligible. See Figure~\ref{fig:del}.

\begin{figure}
\begin{center}
\includegraphics[width=0.8 \columnwidth]{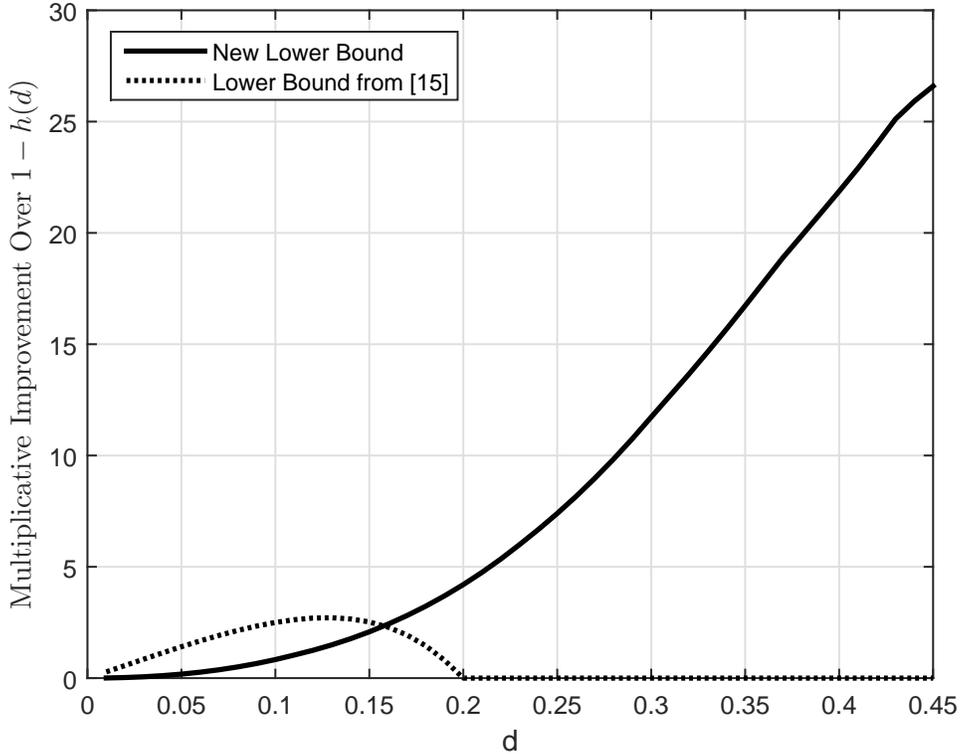}
\end{center}
\caption{The multiplicative improvement factor w.r.t. $1-h(d)$ attained by our lower bound on the mutual information for an i.i.d. uniform input . For comparison, we also plot the improvement the lower bound from~\cite{rd13} attains w.r.t. $1-h(d)$.}
\label{fig:del}
\end{figure}

\vspace{3pt}

\noindent\textit{\underline{Upper Bound for i.i.d Inputs}}

By Theorem \ref{thm:upper} we have in particular that
\begin{align}\label{eq:upper1}
I(\bX; \bY)\leq H(\bY) + \Expt_A\log \Expt_{\bX,\bY} \I(A\sim (\bX,\bY))
\end{align}
Let $\bX$ be an i.i.d. $\textrm{Bern}(q)$ input vector of length $n$ for some $q\leq \tfrac{1}{2}$. It can be shown that the length of $\bY$ is $\Theta\sim \textrm{Binomial}(n,1-d)$, and given its length, $\bY$ is i.i.d. $\textrm{Bern}(q)$. Thus,
\begin{align}\label{eq:HdelY}
\nonumber  \frac{1}{n}H(\bY) &= \frac{1}{n}\left(H(\bY|\Theta) + H(\Theta)\right) \\
  & = (1-d)h(q) + O\left(\frac{\log{n}}{n}\right)
\end{align}

The challenge is thus to evaluate the second term in~\eqref{eq:upper1}, which is given by
\begin{align}
\nonumber  \Expt_A\log \Expt_{\bX,\bY} \I(A\sim (\bX,\bY)) &= \Expt_A\log \Expt_{A'}\Expt_{\bX} \I(A\sim (\bX,A')) \\
  &=\Expt_A\log \Expt_{A'}\Expt_{\bX} \I(A(\bX) = A'(\bX)) \nonumber \\
  & =\Expt_A\log \Expt_{A'}\Pr(A(\bX) = A'(\bX)) \label{eq:upper2}
\end{align}
where $A'\sim P_A$ such that $(\bX,A'(\bX))\sim P_{\bX\bY}$. Note that here $\bX,A,A'$ are mutually independent.

Let us specifically choose $A$ as in Example \ref{ex:del}, namely we identify $A$ with a $\textrm{Bern}(1-d)$ i.i.d. vector of length $n$, and $A(\bX)$ corresponds to sampling $\bX$ in the location chosen by that vector. Asymptotically, we can assume without loss of generality that both $A$ and $A'$ are drawn uniformly over vectors of weight $n(1-d)$. This follows since for any given weight of $A$, the inner expectation w.r.t. $A'$ only increases by replacing the i.i.d. distribution with a uniform distribution over all vectors with the same weight. Furthermore, the outer expectation w.r.t. $A$ is asymptotically dominated by the uniform distribution over vectors of weight $n(1-d)$.

Let us define $S$ to be the action that chooses only the coordinates selected by $A'$ but not by $A$. Let $\overline{S}$ be the complementary action (that chooses only the remaining coordinates). Given any $A'$ and $A$, for any assignment of the values of $\bX$ in the coordinates chosen by $\overline{S}$, there is either a unique assignment $\phi(\overline{S}(\bX))$ of the values of $\bX$ in the coordinates chosen by $S$ that satisfies $A'(\bX)=A(\bX)$, or there is none. In the latter case, we set $\phi(\overline{S}(\bX))$ to an arbitrary value. Thus we can write
\begin{align*}
  \Pr(A(\bX) = A'(\bX)) &= \Pr\left(\bX\in \left\{\bx\in\{0,1\}^n : \I(A'(\bx) = A(\bx)\right\}\right) \\
  & \leq \Pr\left(\bX\in \left\{\bx\in\{0,1\}^n : S(\bx) = \phi(\overline{S}(\bx)) \right\}\right) \\
  & = \Pr\left(S(\bX) = \phi(\overline{S}(\bX))\right)\\
  & = \Expt \Pr\left(S(\bX) = \phi(\overline{S}(\bX)) \mid \overline{S}(\bX)\right)\\
  & \leq  \Expt\max_{\bu\in\{0,1\}^{|S|}}\Pr\left(S(\bX) = \bu  \mid \overline{S}(\bX)\right)\\
  & =  \Expt\max_{\bu\in\{0,1\}^{|S|}}\Pr\left(S(\bX) = \bu\right)\\
  & =  \max_{\bu\in\{0,1\}^{|S|}}\Pr\left(S(\bX) = \bu\right)\\
  & =  (1-q)^{|S|}\\
\end{align*}
Returning to~\eqref{eq:upper2} and using the above, we have
\begin{align*}
\Expt_A\log \Expt_{A'}\Pr(A(\bX) = A'(\bX)) \leq  \Expt_A\log \Expt_{A'}(1-q)^{|S|}
\end{align*}
where the only randomness is in $|S|$, which is a deterministic function of $A$ and $A'$. In particular, $|S|$ is the number of coordinates chosen by $A'$ and not by $A$. Since $A$ and $A'$ were assumed to be uniformly distributed over constant weight vectors of weight $(1-d)n$, then simple counting arguments show that for every action $a$
\begin{align*}
  \Pr(|S| &= \rho(1-d)n | A=a) = \frac{{ (1-d)n \choose (1-\rho)(1-d)n}\cdot {dn \choose \rho(1-d)n}}{{n \choose (1-d)n}} \\
  &\doteq 2^{n\left((1-d)h(\rho) + d\cdot h\left(\rho\frac{1-d}{d}\right) - h(d) \right)}
\end{align*}
Thus, maximizing over feasible values of $\rho$
\begin{align*}
\lim_{n\to\infty}\frac{1}{n}\Expt_A\log \Expt_{A'}\Pr(A(\bX) = A'(\bX)) \leq \max_{0\leq \rho \leq  \frac{d}{1-d}} (1-d)h(\rho) + d\cdot h\left(\rho\frac{1-d}{d}\right) - h(d) + (1-d)\rho\log(1-q)
\end{align*}
Plugging the above in~\eqref{eq:upper1} and using~\eqref{eq:HdelY}, we obtain the bound
\begin{align*}
 \lim_{n\to\infty}\frac{1}{n} I(\bX; \bY) \leq (1-d)h(q) -  h(d) + \max_{0\leq \rho \leq  \frac{d}{1-d}}\Gamma(\rho)
\end{align*}
where
\begin{align*}
\Gamma(\rho)\triangleq (1-d)\left(h(\rho) + \rho\log(1-q)\right) + d\cdot h\left(\rho\frac{1-d}{d}\right).
\end{align*}
We note that the maximization over $\rho$ can be solved directly by differentiation, and the maximizing value is
\begin{align*}
\rho^*=\frac{1-q}{2q(1-d)}\left(\sqrt{1+4d(1-d)\frac{q}{1-q}}-1\right),
\end{align*}
and we therefore have
\begin{align}
\lim_{n\to\infty}\frac{1}{n} I(\bX; \bY) \leq (1-d)h(q) -  h(d)+\Gamma(\rho^*).\label{newUpper}
\end{align}
In the limit of $d\to 1$ it is easy to see that $\rho^*\to d$, and direct substitution into~\eqref{newUpper} reveals that for $q=1/2$ the upper bound is smaller than $(1-d)^2$ for large $d$.
In~\cite{dsv12} it was shown that for an i.i.d. $\mathrm{Bern}(q)$ input process
\begin{align}
\lim_{n\to\infty}\frac{1}{n} I(\bX; \bY) \leq (1-d)\left(h(q)-2dq(1-q) \right).\label{dsvBound}
\end{align}
Our new upper bound is plotted in Figure~\ref{fig:del_upper} for $q=1/2$ along with the upper bound~\eqref{dsvBound} and the trivial upper bound $1-d$. It is seen that for this choice of $q$ our new bound is better than~\eqref{dsvBound} for all deletion probabilities.

We remark that although here we have only applied the bounds from Theorems~\ref{thm:main} and~\ref{thm:upper} for handling deletion channels, we expect a similar approach to yield improved results also for insertion channels.

\begin{figure}
\begin{center}
\includegraphics[width=0.8 \columnwidth]{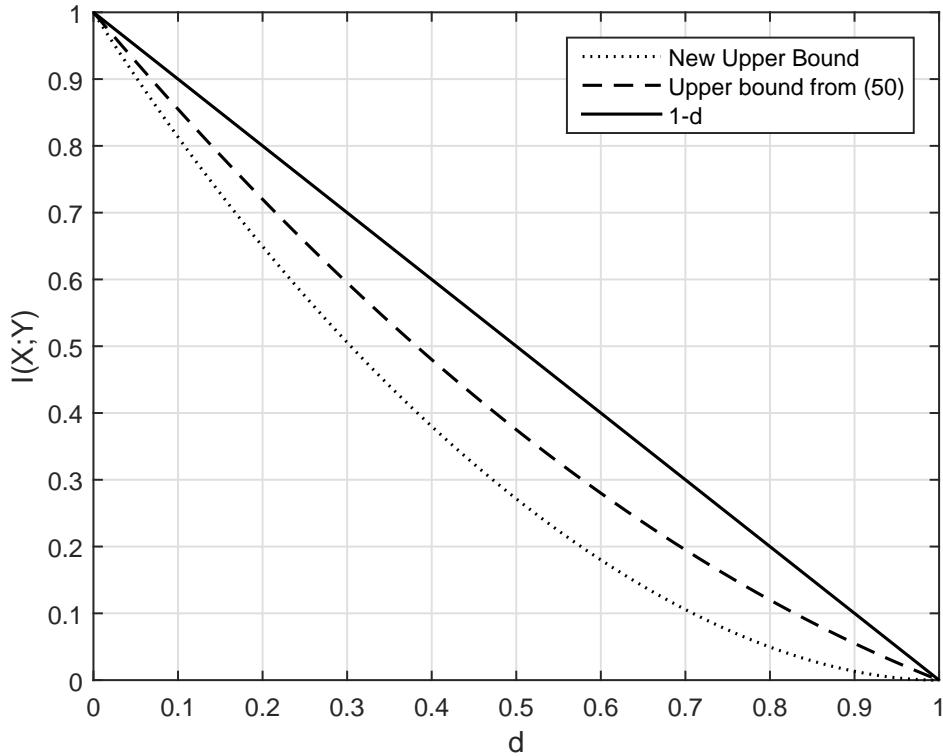}
\end{center}
\caption{Our new upper bound~\eqref{newUpper} plotted for $q=1/2$ along with the upper bound~\eqref{dsvBound} and the trivial upper bound $1-d$.}
\label{fig:del_upper}
\end{figure}

\subsection{Most Informative Boolean Function Conjecture}

Let $\bX$ be an $n$-dimensional binary vector uniformly distributed over $\{0,1\}^n$, and $\bY$ be the output of passing each component of $\bX$ through a binary symmetric channel with crossover probability $\alpha\leq1/2$. Let $f:\{0,1\}^n\to\{0,1\}$ be a boolean function. Following a recent conjecture by Courtade and Kumar~\cite{ck14}, there has been much interest in developing useful upper bounds on $I(f(\bX);\bY)$, where the ultimate goal is to prove that this quantity is maximized by the dictatorship function $f(\bX)=X_i$ for some $i\in[n]$. In this subsection, we apply Theorem~\ref{thm:upper} to derive the following novel upper bound.

\begin{theorem}
Let $\bX,\bZ,\bW\in\{0,1\}^n$ be three statistically independent random vectors, with the entries of $\bX$ i.i.d. $\textrm{Bern}(\tfrac{1}{2})$, and the entries of $\bZ$ and $\bW$ i.i.d. $\textrm{Bern}(\alpha)$. Let $\bY=\bX\oplus\bZ$. For any boolean function $f:\{0,1\}^n\to\{0,1\}$,
\begin{align}
I(\bY;f(\bX))\leq H(f(\bX)) + \Expt_\bW\log \Pr(f(\bX\oplus \bW\oplus \bZ)=f(\bX))
\end{align}
\label{thm:ck}
\end{theorem}

\begin{proof}
Identify the action that maps $\bY$ to $f(\bX)$ with drawing an i.i.d. vector $\bW$ with $\textrm{Bern}(\alpha)$ entries and setting $A(\bY)=f(\bY\oplus \bW)$. The bound~\eqref{eq:upper} reads (discarding the last term which is non-positive)
\begin{align}
I(\bY;f(\bX))&\leq H(f(\bX)) + \Expt_A\log \Expt_{\bY,f(\bX)} \I(A(\bY)=f(\bX))\nonumber\\
&= H(f(\bX)) + \Expt_\bW\log \Expt_{\bY,f(\bX)} \I(f(\bY\oplus\bW)=f(\bX))\nonumber\\
&= H(f(\bX)) + \Expt_\bW\log \Pr(f(\bX\oplus \bW\oplus \bZ)=f(\bX)),\label{IfBound}
\end{align}
as desired
\end{proof}

For a fixed $\bw\in\{0,1\}^n$, let us now express $\Pr(f(\bX\oplus \bw\oplus \bZ)=f(\bX))$. To this end, we use the standard isomorphism $0\to 1$, $1\to -1$, $\oplus\to \cdot$. Under this isomorphism we need to calculate $\Pr(f(\bX\cdot \bw\cdot \bZ)=f(\bX))$, where the products between vectors are taken componentwise. Recall~\cite{ODonnellBook} that $f:\{-1,1\}^n\to\{-1,1\}$ admits the Fourier-Walsh expansion
\begin{align}
f(\bx)=\sum_{S\subseteq [n]}\hat{f}(S)\prod_{i\in S}x_i,\label{eq:foursyn}
\end{align}
where
\begin{align}
\hat{f}(S)\triangleq\mathbb{E}\left(f(\bX)\prod_{i\in S}X_i \right),\label{eq:fouran}
\end{align}
and the expectation is taken w.r.t. to i.i.d. uniform distribution on $\{-1,1\}$. Let $f_\bw(\bX)=f(\bX\cdot \bw)$, and note that it immediately follows from~\eqref{eq:foursyn} that  $\hat{f}_\bw(S)=\hat{f}(S)\prod_{i\in S}w_i$. We have
\begin{align}
\Pr(f(\bX\cdot \bw\cdot \bZ)=f(\bX))&=\Pr(f_\bw(\bX\cdot \bZ)=f(\bX))\nonumber\\
&=\frac{1+\Expt\left(f(\bX)f_\bw(\bX\cdot \bZ)\right) }{2} \nonumber\\
&=\frac{1+\Expt\left(\sum_{S\subseteq [n]}\hat{f}(S)\prod_{i\in S}X_i\sum_{T\subseteq [n]}\hat{f}(T)\prod_{j\in T}X_j Z_j w_j\right) }{2} \nonumber\\
&=\frac{1+\sum_{S\subseteq[n]} \hat{f}^2(S)(1-2\alpha)^{|S|}\prod_{i\in S} w_i }{2},\label{eq:probFour}
\end{align}
where in~\eqref{eq:probFour} we have used the facts that $\mathbb{E}(X_iX_j)=\I(i=j)$ and $\mathbb{E}(Z_i)=(1-2\alpha)$ for any $i,j\in[n]$. Now, substituting~\eqref{eq:probFour} into~\eqref{IfBound} gives the following corollary.
\begin{corollary}
For any boolean $f:\{-1,1\}^n\to\{-1,1\}$,
\begin{align}
I(\bY;f(\bX)&\leq H(f(\bX))-1+\Expt_\bW\log\left(1+\sum_{S\subseteq[n]} \hat{f}^2(S)(1-2\alpha)^{|S|}\prod_{i\in S} W_i\right).\label{Wbound}
\end{align}
where $W_i$ are i.i.d. with $\Pr(W_i=-1)=1-\Pr(W_i=1)=\alpha$.
\label{cor:ck}
\end{corollary}

We note that the upper bound from Theorem~\ref{thm:ck} and Corollary~\ref{cor:ck} are tight for the function $f(\bX)=X_i$. Thus, showing that the dictatorship function maximizes~\eqref{IfBound} or~\eqref{Wbound}, will settle the most informative boolean function conjecture~\cite{ck14}. Unfortunately, our attempts to prove the former were not successful.


\begin{appendix}
\label{app:Zmgf}
Given that $K_{1i}=k_1$ and $K_{2i}=k_2$, we know that the $i$th phrase in the parsing of $\bX$ is of the form
\begin{align}
\underbrace{B\cdots B\overline{B}}_{k_1}\underbrace{\overline{B}\cdots\overline{B}B}_{k_2},\label{ithphrase}
\end{align}
where $B\sim\mathrm{Bern}(\tfrac{1}{2})$ and $\overline{B}\triangleq 1-B$. The r.v. $Z_i$ counts the number of bits in $\bY'$ that were matched by the greedy algorithm to bits in the $i$th phrase of $\bX$. Thus, conditioned on the event $K_{1i}=k_1,K_{2i}=k_2$, the r.v. $Z_i$ counts the number of bits from an i.i.d. uniform sequence (corresponding to the relevant bits in $\bY'$) that are matched by the greedy algorithm to bits in the phrase~\eqref{ithphrase}.

Let $W$ be the event that the first $k_1$ bits of the i.i.d. sequence are equal to $B$. Clearly, $\Pr(W)=2^{-k_1}$ and if $W$ occurs then $Z_i=k_1$. Let $T_1$ be the location of the first occurrence of $\overline{B}$ in the i.i.d. sequence, and let $T'_2$ be the location of the first occurrence of $B$ after $T_1$. Further, let $T_2=T'_2-T_1$. For example, if the sequence of i.i.d. bits is
\begin{align}
B \ B\overline{B} \ \overline{B} \ \overline{B} \ \overline{B} \ \overline{B} \ B\ldots,\nonumber
\end{align}
then $T_1=3$ and $T_2=5$, and if the sequence of i.i.d. bits is
\begin{align}
\overline{B} \ \overline{B} \ B\ldots,\nonumber
\end{align}
then $T_1=1$ and $T_2=2$. We further define the r.v.
\begin{align}
\tilde{T}_2=\begin{cases}
T_2 & T_2\leq k_2\\
k_2-1 & T_2> k_2
\end{cases}.\nonumber
\end{align}
Note that given $\overline{W}$ (the event that $W$ did not occur), we have $Z_i=T_1+\tilde{T}_2$. We have
\begin{align}
\Expt\left(2^{tZ_i} \ | \ {K_{1i}=k_1,K_{2i}=k_2}  \right)&=\Pr(W)\Expt\left(2^{tZ_i} \ | \ {K_{1i}=k_1,K_{2i}=k_2},W  \right)+\Pr(\overline{W})\Expt\left(2^{tZ_i} \ | \ {K_{1i}=k_1,K_{2i}=k_2},\overline{W}\right)\nonumber\\
&=2^{-k_1} 2^{t k_1}+\left(1-2^{-k_1}\right)\Expt\left(2^{t(T_1+\tilde{T}_2)} \ | \ \overline{W}\right)\nonumber\\
&=2^{-k_1} 2^{t k_1}+\left(1-2^{-k_1}\right)\Expt\left(2^{tT_1} \ | \ \overline{W}\right)\Expt\left(2^{t\tilde{T}_2} \right).\label{mgf}
\end{align}
The r.v.s $T_1$ and $T_2$ are statistically independent $\mathrm{Geometric}(\tfrac{1}{2})$, and therefore
\begin{align}
\Pr(T_1=m|\overline{W})=\begin{cases}
\frac{2^{-m}}{1-2^{-k_1}} & 1\leq m\leq k_1-1\\
0 & \text{otherwise}
\end{cases},\nonumber
\end{align}
and
\begin{align}
\Pr(\tilde{T}_2=m)=\begin{cases}
2^{-m} & 1\leq m\leq k_2, m\neq k_2-1\\
2^{-k_2}+2^{-m}\mathds{1}(k_2>1) & m=k_2-1 \\
0 & \text{otherwise}
\end{cases}.\nonumber
\end{align}
This gives
\begin{align}
\Expt\left(2^{tT_1}\right)&=\frac{1}{1-2^{-k_1}}\sum_{m=1}^{k_1-1}2^{-m}2^{tm}\nonumber\\
&=\frac{1}{1-2^{-k_1}}\frac{2^{t-1}}{1-2^{t-1}}\left(1-2^{k_1(t-1)}\right)
\label{mgf1}
\end{align}
and for $k_2>1$
\begin{align}
\Expt\left(2^{t\tilde{T}_2}\right)&=\sum_{m=1}^{k_2}2^{-m}2^{t m}+2^{-k_2}2^{t(k_2-1)}\nonumber\\
&=\frac{2^{t-1}}{1-2^{t-1}}\left(1-2^{k_2(t-1)}\right)+2^{k_2(t-1)-t}.
\label{mgf2}
\end{align}
Note that for $k_2=1$ we have $\Expt\left(2^{t\tilde{T}_2}\right)=\tfrac{1}{2}+\tfrac{1}{2}2^{-t}$, and~\eqref{mgf2} continues to hold. Substituting~\eqref{mgf1} and~\eqref{mgf2} into~\eqref{mgf} yields the desired expression.
\end{appendix}

\bibliographystyle{IEEEtran}
\bibliography{ofer_refs_master}

\end{document}